\documentclass[a4paper,12pt]{article}
\usepackage{a4wide,exscale}
\usepackage{amsmath}
\usepackage{amsfonts}
\usepackage{amssymb}
\usepackage{amsthm}
\usepackage{graphicx}
\usepackage{dsfont}
\usepackage[numbers]{natbib}
\usepackage{verbatim}
\usepackage{subfigure}
\newtheorem{thm}{Theorem}[section]

\newtheorem{coro}[thm]{Corollary}

\newtheorem{defn}[thm]{Definition}
\newtheorem{rem}[thm]{Remark}
\newtheorem{prop}[thm]{Proposition}

\newcommand{\p}[1]{\ensuremath{\mathbb{#1}}}
\newcommand{\gb}[1]{\ensuremath{\gamma_{#1 T}}}

\newcommand{\levy}{{L\'evy }}
\newcommand{\D}[1]{\ensuremath{\Delta_{#1}}}
\newcommand{\tD}[1]{\ensuremath{\Delta^{\pi}_{#1}}}
\def\1{\mathds{1}}
\def\R{\mathbb{R}}

\def\Q{\mathbb{Q}}
\def\E{\mathbb{E}}
\def\N{\mathbb{N}}
\def\d{\,\mathrm{d}}
\def\dd{\mathrm{d}}
\def\half{\frac{1}{2}}
\def\var{\mathrm{Var}}

\def\th{\theta}
\def\e{\mathrm{e}}
\def\b{\beta}

\def\G{\Gamma}
\def\a{\alpha}

\def\s{\sigma}
\def\l{\lambda}
\def\law{\overset{\textrm{law}}{=}}
\def\tp{\mathrm{T}}
\def\F{\mathcal{F}}
\def\lrb{\mathit{LRB}}
\def\lrbc{\mathit{LRB}_{\mathcal{C}}}
\def\lrbd{\mathit{LRB}_{\mathcal{D}}}
\def\i{\mathrm{i}}

\begin{document}

\bibliographystyle{plainnat_ed}

\title{\levy Random Bridges and the Modelling of Financial Information}

\author{Edward Hoyle\thanks{\texttt{ed.hoyle08@imperial.ac.uk}} \thanks{Department of Mathematics, Imperial College London, London SW7 2AZ, UK}
                \and Lane P.~Hughston\footnotemark[2]
                \and Andrea Macrina\thanks{Department of Mathematics, King's College London, London WC2R 2LS, UK}
                        \thanks{Institute of Economic Research, Kyoto University, Kyoto 606-8501, Japan}}
\date{April 2, 2010}
\maketitle


\begin{abstract}
The information-based asset-pricing framework of Brody, Hughston and Macrina (BHM) is extended to include a wider class of models for market information.
In the BHM framework, each asset is associated with a collection of random cash flows.
The price of the asset is the sum of the discounted conditional expectations of the cash flows.
The conditional expectations are taken with respect to a filtration generated by a set of `information processes'.
The information processes carry imperfect information about the cash flows.
To model the flow of information, we introduce in this paper a class of processes which we term \emph{L\'evy random bridges} (LRBs).
This class generalises the Brownian bridge and gamma bridge information processes considered by BHM.
An LRB is defined over a finite time horizon.
Conditioned on its terminal value, an LRB is identical in law to a L\'evy bridge.
We consider in detail the case where the asset generates a single cash flow $X_T$ occurring at a fixed date $T$.
The flow of market information about $X_T$ is modelled by an LRB terminating at the date $T$
with the property that the (random) terminal value of the LRB is equal to $X_T$.
An explicit expression for the price process of such an asset is found by working out the discounted conditional expectation of $X_T$ with respect to the natural filtration of the LRB.
The prices of European options on such an asset are calculated.
\end{abstract}

\section{Introduction and Preliminaries}
In financial markets, the information that traders and investors have about an asset is reflected in its price.
The arrival of new information then leads to changes in asset prices.
The `information-based framework' (or `$X$-factor theory') of Brody, Hughston and Macrina (BHM) isolates the emergence of information, and examines its role as a driver of price dynamics (see \citep{BHM1,BHM2,BHM3,YR2007,MPhD2006,HM2008}).
In the BHM framework, each asset is associated with a collection of random cash flows.
The price of the asset is the sum of the discounted conditional expectations of the cash flows.
The conditional expectations are taken with respect to (i) an appropriate measure, and (ii) the filtration generated by a set of so-called
information processes.
The information processes carry noisy or imperfect market information about the cash flows.
The present paper extends the work of \citep{BHM2} and \citep{BHM3} by introducing a wider class of information processes as a basis for the generation of the market filtration.
The set-up is as follows:

We fix a probability space $(\Omega,\Q,\F)$, and assume that all processes and filtrations under consideration are c\`adl\`ag.
Unless otherwise stated, when discussing a stochastic process we assume that the process takes values in $\R$, begins at time 0, and the filtration is that generated by the process itself.
We work with a finite time horizon $[0,T]$.

\subsection{\levy processes} \label{sec:levy_processes}
This section summarises a few well known results about one-dimensional \levy processes further details of which can be found in \citet{Bert1996} and \citet{Sato1999}.
A \levy process is a stochastically-continuous process that starts from the value 0, and has stationary, independent increments.
An increasing \levy process is called a \emph{subordinator}.
For $\{L_t\}$ a \levy process, its \emph{characteristic exponent} $\Psi:\R\rightarrow\mathbb{C}$ is defined by
\begin{equation}
    \E[\e^{\i\l L_t}]= \exp(-t\Psi(\l)), \qquad \l\in\R.
\end{equation}
The characteristic exponent of a \levy process characterises its law,
and its form is prescribed by the L\'evy-Khintchine formula:
\begin{equation}
    \label{eq:LK}
    \Psi(\lambda)=\i a \lambda+\half \s^2 \lambda^2+\int_{-\infty}^{\infty}(1-\e^{\i x \lambda}+\i x\lambda\1_{\{|x|<1\}})\Pi(\dd x),
\end{equation}
where $a\in\R$, $\s>0$, and $\Pi$ is a measure (the \emph{\levy measure}) on $\R\backslash\{0\}$ such that
\begin{equation}
    \int_{-\infty}^{\infty}(1 \wedge |x|^2)\,\Pi(\dd x)<\infty.
\end{equation}

There are particular subclasses of \levy processes that we shall consider, defined as follows:
\begin{defn}
    Let $\{L_t\}_{0\leq t \leq T}$ and $\{M_t\}_{0\leq t \leq T}$ be \levy processes.
    Then we write
    \begin{enumerate}
        \item
            $\{L_t\}\in\mathcal{C}[0,T]$ if the density of $L_t$ exists for every $t\in (0,T]$,
        \item
            $\{M_t\}\in\mathcal{D}$ if the marginal law of $M_t$ is discrete for some $t>0$.
    \end{enumerate}
\end{defn}
\noindent

\begin{rem}
    If the marginal law of $M_t$ is discrete for some $t>0$, then the marginal law of $M_t$ is discrete for all $t>0$.
    The density of $L_t$ exists if and only if its law is absolutely continuous with respect to the Lebesgue measure.
    In general, the absolute continuity of $L_t$ depends on $t$ \emph{\citep[chap.~5]{Sato1999}}; thus
    $\mathcal{C}[0,T_1]\subseteq \mathcal{C}[0,T_2]$ for $T_1\leq T_2$.
\end{rem}

We reserve the notation $f_t(x)$ to represent the density of $L_t$ for some $\{L_t\}\in\mathcal{C}[0,T]$.
Hence $f_t:\R\rightarrow \R_+$ and $\Q[L_t\in \dd x]=f_t(x)\d x$.
We reserve $Q_t(a)$ to represent the probability mass function of $M_t$ for some $\{M_t\}\in\mathcal{D}$.
We denote the state-space of $\{M_t\}$ by $\{a_i\}\subset\R$.
Hence $Q_t:\{a_i\}\rightarrow [0,1]$ and $\Q[M_t=a_i]=Q_t(a_i)$.
We assume that the sequence $\{a_i\}$ is strictly increasing.

The transition probabilities of \levy processes satisfy the convolution identities
\begin{align}
        f_t(x)&=\int_{-\infty}^{\infty} f_{t-s}(x-y)f_s(y)\d y &&\text{for $\{L_t\}\in\mathcal{C}[0,T]$},
        \\\intertext{and}
        Q_t(a_n)&=\sum_{m=-\infty}^{\infty}Q_{t-s}(a_n-a_m)Q_s(a_m) &&\text{for $\{M_t\}\in\mathcal{D}$},
\end{align}
for $0\leq s<t\leq T$.
These are the Chapman-Kolmogorov equations for the processes $\{L_t\}$ and $\{M_t\}$.

The law of any c\`adl\`ag stochastic process is characterised by its finite-dimensional distributions.
The finite-dimensional densities of $\{L_t\}_{0\leq t\leq T}$ exist and, with the understanding that $x_0=t_0=0$, they are given by
\begin{equation}
    \Q[L_{t_1} \in\dd x_1, \ldots,L_{t_n}\in\dd x_n]=
    \prod_{i=1}^n \left[f_{t_i-t_{i-1}}(x_i-x_{i-1}) \d x_i\right],
\end{equation}
for every $n\in\mathbb{N}_+$, every $0<t_1<\cdots<t_n\leq T$, and every $(x_1,\ldots,x_n)\in\R^n$.
With the understanding that $a_{k_0}=t_0=0$, the finite-dimensional probabilities of $\{M_t\}$ are
\begin{equation}
    \Q[M_{t_1} = a_{k_1}, \ldots,M_{t_n} = a_{k_n}]=\prod_{i=1}^n Q_{t_i-t_{i-1}}(a_{k_i}-a_{k_{i-1}}),
\end{equation}
for every $n\in\mathbb{N}_+$, every $0<t_1<\cdots<t_n$, and every $(k_1,\ldots,k_n)\in\mathbb{Z}^n$.

\subsection{\levy bridges}
A bridge is a stochastic process that is pinned to some fixed point at a fixed future time.
Bridges of Markov processes were constructed and analysed by \citet{FPY1993} in a general setting.
In this section we focus on the bridges of \levy processes in the classes $\mathcal{C}[0,T]$ and $\mathcal{D}$.
In particular we have the following:
\begin{prop}
    \label{prop:LB_Markov}
    The bridges of processes in $\mathcal{C}[0,T]$ and $\mathcal{D}$ are Markov processes.
\end{prop}
\begin{proof}
    We need to the show that the process $\{L_t\}\in\mathcal{C}[0,T]$ is a Markov process when we know that $L_T=x$, for some constant $x$
    such that $0<f_T(x)<\infty$.
    (It will be explained later why the condition that $0<f_T(x)<\infty$ is required to ensure that the law of the bridge process is well defined.)
    In other words, we need to show that
    \begin{equation}
        \Q\left[L_{t} \leq y \,|\, L_{t_1}=x_1,\ldots,L_{t_m}=x_m,L_T=x\right]
                    = \Q\left[L_{t} \leq y \,|\, L_{t_m}=x_m,L_T=x\right],
    \end{equation}
    for all $m \in \mathbb{N}_+$, all $(x_1,\ldots,x_m,y)\in\R^{m+1}$, and all $0\leq t_1<\cdots<t_m<t\leq T$.
    The key property of $\{L_t\}$ that we use is its independent increments.
    Let us write
    \begin{align}
        \Delta_i&=L_{t_i}-L_{t_{i-1}},
        \\\delta_i&=x_i-x_{i-1},
    \end{align}
    for $1\leq i\leq m$, where $t_0=0$ and $x_0=0$.
    Then we have:
    \begin{align}
        &\Q\left[\left. L_{t} \leq y \,\right| L_{t_1}=x_1,\ldots,L_{t_m}=x_m,L_T=x\right] \nonumber
        \\&\quad=\Q\left[\left. L_{t}-L_{t_m} \leq y-x_m \,\right| \Delta_1=\delta_1,\ldots,\Delta_m=\delta_m,L_T-L_{t_m}=x-x_m\right] \nonumber
        \\&\quad=\Q\left[\left. L_{t}-L_{t_m} \leq y-x_m \,\right| L_T-L_{t_m}=x-x_m\right] \nonumber
        \\&\quad=\Q\left[\left. L_{t}-L_{t_m} \leq y-x_m \,\right| L_T-L_{t_m}=x-x_m,L_{t_m}=x_m\right] \nonumber
        \\&\quad=\Q\left[\left. L_{t} \leq y \,\right| L_T=x,L_{t_m}=x_m\right].
    \end{align}
    The proof for processes in class $\mathcal{D}$ is similar.
\end{proof}

Let $\{L_t\}\in\mathcal{C}[0,T]$, and let $\{L^{(z)}_{tT}\}_{0 \leq t\leq T}$ be an $\{L_t\}$-bridge to the value $z\in\R$ at time $T$.
For the transition probabilities of the bridge process to be well defined, we require that $0<f_T(z)<\infty$.
By the Bayes theorem we have
\begin{align}
    \Q\left[L^{(z)}_{tT}\in \dd y \left|\, L^{(z)}_{sT}=x \right.\right]&=
                \Q\left[L_{t}\in \dd y \left|\, L_{s}=x, L_T=z \right.\right] \nonumber
    \\ &=\frac{\Q\left[L_{t}\in \dd y, L_T\in \dd z \left|\, L_{s}=x \right.\right]}{\Q\left[L_T\in\dd z \left|\, L_{s}=x \right.\right]} \nonumber
    \\ &=\frac{f_{t-s}(y-x)f_{T-t}(z-y)}{f_{T-s}(z-x)} \d y,
\end{align}
for $0\leq s<t<T$.
We define the marginal bridge density $f_{tT}(y;z)$ by
\begin{equation}
    \label{eq:def_ftT}
    f_{tT}(y;z)=\frac{f_t(y)f_{T-t}(z-y)}{f_T(z)}.
\end{equation}
In this way
\begin{equation}
    \Q\left[L^{(z)}_{tT}\in \dd y \left|\, L^{(z)}_{sT}=x \right.\right]=f_{t-s,T-s}(y-x;z-x) \d y.
\end{equation}
The condition $0<f_T(z)<\infty$ is enough to ensure that
\begin{equation}
    y\mapsto f_{t-s,T-s}(y-L^{(z)}_{sT};z-L^{(z)}_{sT})
\end{equation}
is a well defined density for almost every value of $L^{(z)}_{sT}$.
To see this, note that
\begin{align}
    & \int_{-\infty}^{\infty}\int_{-\infty}^{\infty} f_{t-s,T-s}(y-x;z-x) \, \Q\left[L^{(z)}_{sT}\in\dd x\right]  \dd y \nonumber
    \\ &\qquad=\int_{-\infty}^{\infty}\int_{-\infty}^{\infty} f_{t-s,T-s}(y-x;z-x) f_{s,T}(x;z) \d x  \d y \nonumber
    \\ &\qquad=\int_{-\infty}^{\infty}\frac{f_{T-t}(z-y)}{f_T(z)}\int_{-\infty}^{\infty} f_{t-s}(y-x) f_{s}(x) \d x  \d y \nonumber
    \\ &\qquad=\frac{1}{f_T(z)} \int_{-\infty}^{\infty} f_{T-t}(z-y) f_t(y) \d y=1. \label{eq:ImADensity}
\end{align}
From (\ref{eq:ImADensity}) it follows that
\begin{equation}
    \label{eq:check_den}
    \Q\left[ \int_{-\infty}^{\infty} f_{t-s,T-s}(y-L^{(z)}_{sT};z-L^{(z)}_{sT}) \d y=1\right]=1.
\end{equation}

Let $\{M_t\}\in\mathcal{D}$, and let $\{M^{(k)}_{tT}\}_{0 \leq t\leq T}$ be an $\{M_t\}$-bridge to the value $a_k$ at time $T$,
so $\Q[M^{(k)}_{TT}=a_k]=1$.
For the transition probabilities of the bridge to be well defined, we require that $\Q[M_T=a_k]=Q_T(a_k)>0$.
Then the Bayes theorem gives
\begin{align}
    \Q\left[M^{(k)}_{tT}=a_j \left|\, M^{(k)}_{sT}=a_i \right.\right]&=
                \Q\left[M_{t}=a_j \left|\, M_{s}=a_i, M_T=a_k \right.\right] \nonumber
    \\ &=\frac{\Q\left[M_{t}=a_j, M_T=a_k \left|\, M_{s}=a_i \right.\right]}{\Q\left[M_T=a_k \left|\, M_{s}=a_i \right.\right]} \nonumber
    \\ &=\frac{Q_{t-s}(a_j-a_i)Q_{T-t}(a_k-a_j)}{Q_{T-s}(a_k-a_i)}, \label{eq:ratio}
\end{align}
for $0\leq s<t<T$.
Note that if $Q_T(a_k)=0$, then the ratio (\ref{eq:ratio}) is not well defined when $s=0$.

%
%

\section{\levy random bridges} \label{chap:LRB}
The idea of information-based asset pricing is to model the flow of information in financial markets and hence to construct the market filtration explicitly.
Let $X_T$ be a random variable (a market factor), with a given \emph{a priori} distribution.
The value of $X_T$ will be revealed to the market at time $T$.
We wish to construct an information process $\{\xi_{tT}\}$ such that $\xi_{TT}=X_T$.
We can then use the filtration generated by $\{\xi_{tT}\}$ to model the information that market participants have about $X_T$.
One problem to overcome is how to ensure that the marginal law of $\xi_{TT}$ is the \emph{a priori} law of $X_T$.

Two explicit forms for the information process have been considered in the literature.
The first is
\begin{equation}
    \label{eq:BB}
    \xi_{tT}=\frac{t}{T}X_T+\beta_{tT} \qquad (0\leq t \leq T),
\end{equation}
where $\{\beta_{tT}\}_{0\leq t\leq T}$ is a Brownian bridge starting and ending at the value 0 (see \citep{BHM1,BDFH2008,BHM2,HM2008,MPhD2006,YR2007}).
The second is
\begin{equation}
    \xi_{tT}=X_T \gb{t} \qquad (0\leq t \leq T),
\end{equation}
where $X_T>0$ and $\{\gb{t}\}_{0 \leq t\leq T}$ is a gamma bridge starting at the value 0 and ending at the value 1 (see \citep{BHM3}).
These forms share the property that each is identical in law to a \levy process conditioned to have the \emph{a priori} law of $X_T$ at time $T$.
The Brownian bridge information process is identical in law to a conditioned Brownian motion, and the gamma bridge information process is identical in law to a conditioned gamma process.

With this as motivation, in this section we define a class of processes that we call \levy random bridges (LRBs).
An LRB is identical in law to a \levy process conditioned to have a prespecified marginal law at $T$.
Later we shall use LRBs as information processes in information-based models.

\subsection{Defining LRBs}
An LRB can be described as a process whose bridge laws are \levy bridge laws.
In the definitions below we define LRBs by reference to their finite-dimensional distributions rather than as conditioned \levy processes.
This proves convenient in future calculations.

\begin{defn}
    \label{def:LRB}
    We say that the process $\{L_{tT}\}_{0\leq t\leq T}$ has law the $\lrb_{\mathcal{C}}([0,T],\{f_t\},\nu)$ if the following are satisfied:
    \begin{enumerate}       \item $L_{TT}$ has marginal law $\nu$.
        \item There exists a \levy process $\{L_t\}\in\mathcal{C}[0,T]$ such that $L_t$ has density $f_t(x)$ for all $t\in(0,T]$.
        \item \label{item:measurecond} $\nu$ concentrates mass where $f_T(z)$ is positive and finite, i.e.~$0<f_T(z)<\infty$ for $\nu$-a.e.~$z$.
        \item \label{item:condLevy} For every $n\in\mathbb{N}_+$, every $0<t_1<\cdots<t_n<T$, every $(x_1,\ldots,x_n)\in\R^n$, and $\nu$-a.e.~$z$, we have
                        \begin{equation*}
                                            \Q\left[L_{t_1,T}\leq x_1,\ldots,L_{t_n,T} \leq x_n \left|\, L_{TT} = z \right.\right]
                                             =\Q\left[L_{t_1}\leq x_1,\ldots,L_{t_n} \leq x_n \left|\, L_{T} = z \right.\right].
                        \end{equation*}
    \end{enumerate}
\end{defn}

\begin{defn}
    We say that the process $\{M_{tT}\}_{0\leq t\leq T}$ has law the $\lrb_{\mathcal{D}}([0,T],\{Q_t\},P)$ if the following are satisfied:
    \begin{enumerate}
        \item $M_{TT}$ has probability mass function $P$.
        \item There exists a \levy process $\{M_t\}\in\mathcal{D}$ such that $M_t$ has marginal probability mass function $Q_t(a)$ for all $t\in(0,T]$.
        \item The law of $M_{TT}$ is absolutely continuous with respect to the law of $M_T$, i.e.~%
        \begin{equation*}
            \text{if } P(a)>0 \text{ then } Q_T(a)>0.
        \end{equation*}
        \item For every $n\in\mathbb{N}_+$, every $0<t_1<\cdots<t_n<T$, every $(k_1,\ldots,k_n)\in\mathbb{Z}^n$, and every $b$ such that $P(b)>0$, we have
                        \begin{multline*}
                                \Q\left[M_{t_1,T}= a_{k_1},\ldots,M_{t_n,T} = a_{k_n} \left|\, M_{TT} = b \right.\right]
                                            =\\\Q\left[M_{t_1}= a_{k_1},\ldots,M_{t_n} = a_{k_n} \left|\, M_{T} = b \right.\right].
                        \end{multline*}
    \end{enumerate}
\end{defn}

\begin{defn}
For a fixed time $s<T$, if the law of the process $\{\eta_{s+t}\}_{0\leq t \leq T-s}$ is of the type \mbox{$\lrb_{\mathcal{C}}([0,T-s],\,\cdot\,,\,\cdot\,)$}, resp.~\mbox{$\lrb_{\mathcal{D}}([0,T-s],\,\cdot\,,\,\cdot\,)$},
then we say that $\{\eta_{t}\}_{s\leq t \leq T}$ has law $\lrb_{\mathcal{C}}([s,T],\,\cdot\,,\,\cdot\,)$, resp.~$\lrb_{\mathcal{D}}([s,T],\,\cdot\,,\,\cdot\,)$.
\end{defn}
If the law of a process is one of the $\lrb$-types defined above, then we say that it is a \levy random bridge (LRB).

\subsection{Finite-dimensional distributions} \label{sec:LRB_FDD}
For the rest of this section we assume that $\{L_{tT}\}$ and $\{M_{tT}\}$ are LRBs with laws $\lrbc([0,T],\{f_t\},\nu)$
and $\lrbd([0,T],\{Q_t\},P)$, respectively.
We also assume that $\{L_t\}$ is a \levy process such that $L_t$ has density $f_t(x)$ for $t\leq T$, and
$\{M_t\}$ is a \levy process such that $M_t$ has probability mass function $Q_t(a_i)$ for $t\leq T$.

The finite dimensional distributions of $\{L_{tT}\}$ are given by
\begin{equation}
    \label{eq:LRBlaw}
     \Q\left[ L_{t_1,T}\in \dd x_1,\ldots, L_{t_n,T}\in\dd x_n, L_{TT}\in\dd z \right]=
    \prod_{i=1}^{n} \left[f_{t_i-t_{i-1}}(x_i-x_{i-1}) \d x_i\right] \psi_{t_{n}}(\dd z;x_{n}),
\end{equation}
where the (un-normalised) measure $\psi_t(\dd z;\xi)$ is given by
\begin{align}
    \psi_0(\dd z;\xi)&=\nu(\dd z),
    \\\psi_t(\dd z;\xi)&=\frac{f_{T-t}(z-\xi)}{f_T(z)}\nu(\dd z), \label{eq:psit}
\end{align}
for $0<t<T$.
It follows from the definition of $\lrbc([0,T],\{f_t\},\nu)$ and (\ref{eq:check_den}) that
\begin{equation}
    f_{tT}(x;z)=\frac{f_t(x)f_{T-t}(z-x)}{f_T(z)}
\end{equation}
is a well-defined density (as a function of $x$) for $t<T$ and $\nu$-a.e.~$z$.
Then from (\ref{eq:LRBlaw}) the marginal law of $L_{tT}$ is given by
\begin{align}
    \Q[L_{tT}\in \dd x]&=f_t(x) \psi_t(\R;x) \d x \nonumber
    \\&=\int_{z=-\infty}^{\infty} f_{tT}(x;z) \,\nu(\dd z) \d x.
\end{align}
Hence the density of $L_{tT}$ exists for $t<T$, and
\begin{equation}
    0\leq \psi_t(\R;x)<\infty \qquad \text{for Lebesgue-a.e.~$x\in\mathrm{Support}(f_t)$.}
\end{equation}
In particular, we have
\begin{align}
    \label{eq:density_conditions}
    &0< \psi_t(\R;L_{tT})<\infty &&\text{and} &&0< f_{T-t}(x-L_{tT})<\infty
\end{align}
for a.e.~value of $L_{tT}$.
If $\nu(\{z\})=1$ for some point $z\in\R$, i.e.~$\Q[L_{TT}=z]=1$, then $\{L_{tT}\}$ is a \levy bridge.
If $\nu(\dd z)=f_T(z)\d z$, then $\{L_{tT}\}\law\{L_t\}$ for $t\in[0,T]$.

In the discrete case, the finite-dimensional probabilities of $\{M_{tT}\}$ are
\begin{equation}
     \Q\left[ M_{t_1,T}= a_{k_1},\ldots, M_{t_n,T}= a_{k_n}, M_{TT}= z \right]=
    \\ \prod_{i=1}^{n} \left[Q_{t_i-t_{i-1}}(a_{k_i}-a_{k_{1-1}}) \right] \phi_{t_{n}}(z;a_{k_n}),
\end{equation}
where the function $\phi_t(z;\xi)$ is given by
\begin{align}
    \phi_0(z;\xi)&=P(z),
    \\\phi_t(z;\xi)&=\frac{Q_{T-t}(z-\xi)}{Q_T(z)}P(z),
\end{align}
for $0<t<T$.
If $P$ is identical to $Q_T$, then $\{M_{tT}\}\law\{M_t\}$ for $t\in[0,T]$.

The existing literature on information-based asset pricing exploits special properties Brownian and gamma bridges.
See \citet{EY2004} for insights into how remarkable these bridges are.
The methods we use do not require special properties of particular \levy bridges.
However, we use the Brownian and gamma cases as examples, and the results we obtain agree with previous work.

Many of the results that follow are proved for the LRB $\{L_{tT}\}$, which has a continuous state-space.
Analogous results are provided for the discrete state-space process $\{M_{tT}\}$; details of proofs are omitted since they are similar to the continuous case.

\subsection{LRBs as conditioned \levy processes}
It is useful to interpret an LRB as a \levy process conditioned to have a specified marginal law $\nu$ at time $T$.
Suppose that the random variable $Z$ has law $\nu$, then:
\begin{align}
 &\Q\left[L_{t_1} \in\dd x_1, \ldots,L_{t_{n}}\in\dd x_{n}, L_{T}\in\dd z \left|\, L_T=Z \right.\right] \nonumber
 \\ &\qquad=     \Q\left[L_{t_1} \in\dd x_1, \ldots,L_{t_{n}}\in\dd x_{n} \left|\, L_T=z\right.\right]  \nu( \dd z) \nonumber
 \\ &\qquad=    \frac{f_{T-t_{n-1}}(z-x_{n-1})}{f_T(z)} \, \prod_{i=1}^{n} \left[f_{t_i-t_{i-1}}(x_i-x_{i-1}) \d x_i\right]\nu( \dd z).
\end{align}
Hence the conditioned \levy process has law $\lrbc([0,T],\{f_t\},\nu)$.

\subsection{The Markov property}
In this section we show that LRBs are Markov processes.
The Markov property is a key tool in the application of LRBs to information-based asset pricing.
As will be seen below, the Markov property of an LRB follows from the Markov property from the associated \levy bridge processes.

\subsubsection{Continuous state-space}
\begin{prop}
    \label{prop:markov}
    The process $\{L_{tT}\}_{0\leq t \leq T}$ is a Markov process with transition law
    \begin{equation}
        \label{eq:LRBtranslaw}
            \begin{aligned}
                \Q[L_{tT} \in \dd y \,|\, L_{sT}=x]&=\frac{\psi_t(\R;y)}{\psi_s(\R;x)} f_{t-s}(y-x)\d y,
                \\\Q[L_{TT} \in \dd y \,|\, L_{sT}=x]&=\frac{\psi_s(\dd y;x)}{\psi_s(\R;x)},
            \end{aligned}
    \end{equation}
    for $0\leq s<t<T$.
\end{prop}
\begin{proof}
    To show that $\{L_{tT}\}$ is Markov, it is sufficient to show that
    \begin{equation}
        \Q\left[L_{tT} \leq y \,|\, L_{t_1,T}=x_1,\ldots,L_{t_m,T}=x_m\right]
                    = \Q\left[L_{tT} \leq y \,|\, L_{t_m,T}=x_m\right],
    \end{equation}
    for all $m \in \mathbb{N}_+$, all $(x_1,\ldots,x_m,y)\in\R^{m+1}$, and all $0\leq t_1<\cdots<t_m<t\leq T$.
    When $t=T$ we apply the Bayes theorem to (\ref{eq:LRBlaw}) and obtain
    \begin{equation}
         \Q\left[\left. L_{TT}\in\dd y  \,\right|  L_{t_1,T}=x_1,\ldots,L_{t_m,T}=x_m\right]=\frac{\psi_{t_m}(\dd y;x_m)}{\psi_{t_m}(\R;x_m)}.
    \end{equation}
    We need now only consider the case $t<T$.
    Proposition \ref{prop:LB_Markov} shows that \levy bridges are Markov processes; therefore,
    \begin{equation}
        \label{eq:LBMarkov}
         \Q\left[L_{t}\leq y \,|\, L_{t_1}=x_1,\ldots,L_{t_m}=x_m,L_T=x\right]=
                     \Q\left[L_{t}\leq y \,|\, L_{t_m}=x_m,L_T=x\right].
    \end{equation}
    It is straightforward by Definition \ref{def:LRB} part \ref{item:condLevy} to show that LRBs are Markov processes.
    Indeed we have:
    \begin{align}
        &\Q\left[\left. L_{tT} \leq y \,\right| L_{t_1,T}=x_1,\ldots,L_{t_m,T}=x_m\right] \nonumber
        \\  &\qquad=\int_{-\infty}^{\infty} \Q\left[\left. L_{tT} \leq y \,\right| L_{t_1,T}=x_1,\ldots,L_{t_m,T}=x_m,L_{T,T}=x\right]\nu(\dd x) \nonumber
        \\  &\qquad=\int_{-\infty}^{\infty} \Q\left[\left. L_{t} \leq y \,\right| L_{t_1}=x_1,\ldots,L_{t_m}=x_m,L_{T}=x\right]\nu(\dd x) \nonumber
        \\  &\qquad=\int_{-\infty}^{\infty} \Q\left[\left. L_{t} \leq y \,\right| L_{t_m}=x_m,L_{T}=x\right]\nu(\dd x) \nonumber
        \\  &\qquad=\int_{-\infty}^{\infty} \Q\left[\left. L_{tT} \leq y \,\right| L_{t_m,T}=x_m,L_{T,T}=x\right]\nu(\dd x) \nonumber
        \\  &\qquad= \Q\left[\left. L_{tT} \leq y \,\right| L_{t_m,T}=x_m\right].
    \end{align}
    The form of the transition law of $\{L_{tT}\}$ appearing in (\ref{eq:LRBtranslaw}) follows from (\ref{eq:LRBlaw}).
\end{proof}

\bigskip

\noindent{\bf Example.}
In the Brownian case we set
\begin{equation}
    \label{eq:NormalDensity}
    f_t(z)=\frac{1}{\sqrt{2 \pi t}} \exp\left[-\frac{z^2}{2t} \right]
\end{equation}
for $t>0$.
Thus $f_t(x)$ is the marginal density of a standard Brownian motion at time $t$.
Then we have
\begin{equation}
    \Q[L_{tT} \in \dd y \,|\, L_{sT}=x]=
    \sqrt{\frac{T-s}{T-t}}\frac{\int_{-\infty}^{\infty}\e^{-\half\left[\frac{(z-y)^2}{T-t}-\frac{z^2}{T}\right]}\,\nu(\dd z)}
                {\int_{-\infty}^{\infty}\e^{-\half\left[\frac{(z-x)^2}{T-s}-\frac{z^2}{T}\right]}\,\nu(\dd z)} \frac{\e^{-\half \frac{(y-x)^2}{t-s}}}{\sqrt{2\pi (t-s)}}\d y,
\end{equation}
and
\begin{align}
    \Q[L_{TT} \in \dd y \,|\, L_{sT}=x] &=\frac{\e^{-\half\left[\frac{(y-x)^2}{T-s}-\frac{y^2}{T}\right]}\,\nu(\dd y)}
                {\int_{-\infty}^{\infty}\e^{-\half\left[\frac{(z-x)^2}{T-s}-\frac{z^2}{T}\right]}\,\nu(\dd z)} \nonumber
    \\&=\frac{\e^{\frac{1}{T-s}\left[xy-\half \frac{s}{T}y^2\right]}\,\nu(\dd y)}
                {\int_{-\infty}^{\infty}\e^{\frac{1}{T-s}\left[xz-\half \frac{s}{T}z^2\right]}\,\nu(\dd z)}.
\end{align}

\bigskip

\noindent{\bf Example.}
In the gamma case we consider a one-parameter family of processes indexed by $m>0$. We set
\begin{equation}
    \label{eq:gamma_den_2}
    f_t(z)=\1_{\{z>0\}} \frac{z^{mt-1}}{\G[mt]} \e^{-z},
\end{equation}
where $\G[z]$ is the gamma function, defined as usual for $x>0$ by
\begin{equation}
    \G[x]=\int_0^\infty u^{x-1} \e^{-u} \d u.
\end{equation}
These densities are the increment densities of the gamma process with mean $m$ and variance $m$ at time $t=1$ (see \citet{BHM3}).
Then
\begin{multline}
    \Q[L_{tT} \in \dd y \,|\, L_{sT}=x] \\
        =\frac{\1_{\{y>x\}}}{\mathrm{B}[m(T-t),m(t-s)]}\frac{\int_{y}^{\infty}(z-y)^{m(T-t)-1}z^{1-mT}\,\nu(\dd z)}{\int_{x}^{\infty}(z-x)^{m(T-s)-1}z^{1-mT}\,\nu(\dd z)}
            (y-x)^{m(t-s)-1} \d y,
\end{multline}
and
\begin{equation}
    \Q[L_{TT} \in \dd y \,|\, L_{sT}=x]
        =\frac{\1_{\{y>x\}}(y-x)^{m(T-s)-1}y^{1-mT}\,\nu(\dd y)}{\int_{x}^{\infty}(z-x)^{m(T-s)-1}z^{1-mT}\,\nu(\dd z)}.
\end{equation}
Here $\mathrm{B}[\a,\b]$ is the beta function, defined as usual for $\a>0$ and $\b>0$ by
\begin{equation}
    \mathrm{B}[\a,\b]=\int_0^1 x^{\a-1} (1-x)^{\b-1} \d x=\frac{\G[\a]\G[\b]}{\G[\a+\b]}.
\end{equation}

\subsubsection{Discrete state-space}
The analogous result to Proposition \ref{prop:markov} for the discrete case is provided below---the proof is similar.
\begin{prop}
    The process $\{M_{tT}\}_{0\leq t\leq T}$ has the Markov property, with transition probabilities given by
    \begin{equation}
        \label{eq:discreteProb}
        \begin{aligned}
        \Q\left[M_{tT} = a_j \left|\, M_{sT}=a_i\right.\right]
            &=\frac{\sum_{k=-\infty}^{\infty} \phi_t(a_k;a_j)}{\sum_{k=-\infty}^{\infty} \phi_s(a_k;a_i)} Q_{t-s}(a_j-a_i),
        \\\Q\left[M_{TT} =a_j \left|\, M_{sT}=a_i\right.\right]&=\frac{\phi_s(a_j;a_i)}{\sum_{k=-\infty}^{\infty} \phi_s(a_k;a_i)},
        \end{aligned}
    \end{equation}
    for $0\leq s<t<T$.
\end{prop}

\subsection{Conditional terminal distributions}
Let $\{\F^L_t\}$ and $\{\F^{M}_t\}$ be the filtrations generated by $\{L_{tT}\}$ and $\{M_{tT}\}$, respectively.
\begin{defn}
Let $\nu_s$ to be the $\F^L_s$-conditional law of the terminal value $L_{TT}$,
and let $P_s$ to be the $\F^M_s$-conditional probability mass function of the terminal value $M_{TT}$.
\end{defn}
We have $\nu_0(A)=\nu(A)$, and $P_0(a)=P(a)$.
Furthermore, when $s>0$, it follows from the results of the previous section that
\begin{align}
    \label{eq:C_terminal}
    \nu_s(\dd z)&= \frac{\psi_s(\dd z; L_{sT})}{\psi_s(\R;L_{sT})},
    \\ \intertext{and}
    \label{eq:D_terminal}
    P_s(a_k)&=\frac{\phi_s(a_k;M_{sT})}{\sum_{j=-\infty}^{\infty} \phi_s(a_j;M_{sT})}.
\end{align}
When the \emph{a priori} $q$th moment of $L_{TT}$ is finite, the $\F^L_s$-conditional $q$th moment is finite and given by
\begin{equation}
    \label{eq:LpMart}
    \int_{-\infty}^{\infty} |z|^q \, \nu_s(\dd z).
\end{equation}
Similarly, when the \emph{a priori} $q$th moment of $M_{TT}$ is finite, the $\F^M_s$-conditional $q$th moment is finite and given by
\begin{equation}
    \label{eq:MpMart}
    \sum_{k=-\infty}^{\infty} |a_k|^q \, P_s(a_k).
\end{equation}
When they are finite, the quantities in (\ref{eq:LpMart}) and (\ref{eq:MpMart}) are martingales
with respect to $\{\F^L_t\}$ and $\{\F^M_t\}$, respectively.
If $q\in\p{Z}$ then $\int |z|^q \, \nu(\dd z)<\infty$ ensures that $\int z^q \, \nu(\dd z)$ is a martingale,
and $\sum |a_k|^q P(a_k)<\infty$ ensures that $\sum a_k^q \,P(a_k)$ is a martingale.

When the terminal law $\nu$ admits a density, we denote it by $p(z)$, i.e.~$\nu(\dd z)=p(z)\d z$.
In this case the $L_{tT}$-conditional density of $L_{TT}$ exists, and we denote it by
\begin{equation}
    p_t(z)=\frac{\nu_t(\dd z)}{\dd z}=\frac{f_{T-t}(z-L_{tT})p(z)}{\psi_t(\R;L_{tT})f_T(z)}.
\end{equation}

\subsection{Measure changes} \label{sec:measure}
In this section we assume that there exists a measure $\p{L}$ under which $\{L_{tT}\}$ is a \levy process, and that the
density of $L_{tT}$ is $f_t(x)$.
Writing $\psi_t=\psi_t(\R;L_{tT})$, we can show that $\{\psi_t\}_{0\leq t<T}$ is an $\p{L}$-martingale (with respect to the filtration generated by $\{L_{tT}\}$).
In particular, for times $0\leq s<t$ we have
\begin{align}
    \E_{\p{L}}\left[\psi_t \left|\, \F^{L}_s \right.\right]
    &=\E_{\p{L}}\left[\left.\int_{-\infty}^{\infty}\frac{f_{T-t}(z-L_{tT})}{f_T(z)}\, \nu(\dd z)\,\right|\F^{L}_s  \right] \nonumber
    \\&=\E_{\p{L}}\left[\left.\int_{-\infty}^{\infty}\frac{f_{T-t}(z-L_{sT}-(L_{tT}-L_{sT}))}{f_T(z)}\, \nu(\dd z)\,\right|L_{sT} \right] \nonumber
    \\&=\int_{y=-\infty}^{\infty}\int_{z=-\infty}^{\infty}\frac{f_{T-t}(z-L_{sT}-y)}{f_T(z)}\, \nu(\dd z)\,f_{t-s}(y) \d y \nonumber
    \\&=\int_{z=-\infty}^{\infty}\frac{1}{f_T(z)}\int_{y=-\infty}^{\infty}f_{T-t}(z-L_{sT}-y)f_{t-s}(y)\d y \, \nu(\dd z) \nonumber
    \\&=\int_{z=-\infty}^{\infty}\frac{f_{T-s}(z-L_{sT})}{f_T(z)} \, \nu(\dd z) \nonumber
    \\ &=\psi_s.
\end{align}
Since $\psi_0=1$, we can define a probability measure $\p{L}^{\mathrm{rb}}$ by the Radon-Nikod\'ym derivative
\begin{equation}
    \left.\frac{\dd \p{L}^{\mathrm{rb}}}{\dd \p{L}}\right|_{\F^{L}_t}=\psi_t \qquad \text{for $0\leq t <T$}.
\end{equation}
It was noted in Section \ref{sec:LRB_FDD} that $0<\psi_t<\infty$, so $\p{L}^{\mathrm{rb}}$ is equivalent to $\p{L}$ for $t<T$.
For $0\leq s<t<T$, the transition law of $\{L_{tT}\}$ under $\p{L}^{\mathrm{rb}}$ is
\begin{align}
    \p{L}^{\mathrm{rb}}\left[L_{tT}\in\dd y \left|\, \F_{s}^{L}\right.\right]
         &=\E_{\p{L}^{\mathrm{rb}}}\left[\1_{\{L_{tT}\in\dd y\}}  \left|\, \F_{s}^L\right.\right] \nonumber
    \\ &=\psi_s^{-1} \,\E_{\p{L}} \left[\psi_t \1_{\{L_{tT}\in\dd y\}}  \left|\, L_{sT}\right.\right] \nonumber
    \\ &=\psi_s^{-1} \int_{-\infty}^{\infty} \frac{f_{T-t}(z-y)}{f_T(z)} \, \nu(\dd z) \, f_{t-s}(y-L_{sT})\d y \nonumber
    \\ &= \frac{\psi_t(\R;y)}{\psi_s(\R;L_{sT})} f_{t-s}(y-L_{sT})\d y.
\end{align}
We see that $\{L_{tT}\}_{0\leq t<T}$ is a Markov process under the measure $\p{L}^{\mathrm{rb}}$.
Furthermore, by virtue of Proposition \ref{prop:markov}, $\{L_{tT}\}$ is an LRB with law $\lrbc([0,T],\{f_t\},\nu)$.

We can restate this result with reference to the measure $\Q$ as the following:
\begin{prop}
    Let $\p{L}$ be defined by
    \begin{equation}
        \left. \frac{\dd \p{L}}{\dd \Q} \right|_{\F^L_t}= \psi_t(\R;L_{tT})^{-1}
    \end{equation}
    for $t\in[0,T)$.
    Then $\p{L}$ is a probability measure.
    Under $\p{L}$, $\{L_{tT}\}_{0\leq t <T}$ is a \levy process, and $L_{tT}$ has density $f_t(x)$.
\end{prop}

In the case of a discrete state space a similar result is obtained.
\begin{prop}
    Let $\p{L}$ be defined by
    \begin{equation}
        \left. \frac{\dd \p{L}}{\dd \Q} \right|_{\F^M_t}= \left[ \sum_{k=-\infty}^{\infty} \phi_t(a_k;M_{tT}) \right]^{-1}
    \end{equation}
    for $t\in[0,T)$.
    Then $\p{L}$ is a probability measure.
    Under $\p{L}$, $\{M_{tT}\}_{0\leq t <T}$ is a \levy process, and $M_{tT}$ has mass function $Q_t(a)$.
\end{prop}

\subsection{Dynamic consistency}
In this section we show that LRBs possess the so-called dynamic consistency property.
For $\{L_{tT}\}$, this property means the process $\{\eta_{t}\}$ defined by setting
\begin{equation}
    \label{eq:define_eta}
    \eta_{t}=L_{tT}-L_{sT} \qquad (s\leq t\leq T)
\end{equation}
is an LRB for fixed $s$ and $L_{sT}$ given.
Defining the filtration $\{\F^{\eta}_t\}$ by
\begin{equation}
    \F^{\eta}_t=\s\left(L_{sT}, \{\eta_u\}_{s\leq u\leq t}  \right),
\end{equation}
we see that
\begin{equation}
    \Q\left[ F\left(\{L_{uT}\}_{s\leq u \leq T}\right) \left|\, \F^{\eta}_t \right.\right]
        =   \Q\left[ F\left(\{L_{uT}\}_{s\leq u \leq T}\right) \left|\, \F^{L}_t \right.\right],
\end{equation}
for $0\leq s<t<T$ and $F$ an arbitrary measurable functional.
Suppose two market participants, trader A and trader B, watch the evolution of $\{L_{tT}\}$; trader A watching from $t=0$ and trader B watching from $t=s$.
The filtration of trader A, $\{\F^{L}_t\}$, is larger than the filtration of trader B, $\{\F^{\eta}_t\}$, but they have a common view of the future evolution of $\{L_{tT}\}$.
This is the Markov property.
The dynamic consistency property is stronger.
It states that the filtration of trader B can be regarded as being generated by an LRB, in this case $\{\eta_t\}$, plus some information about the current state of the world, in this case $L_{sT}$.

Later we shall model the market filtration as being generated by a set of LRBs.
Through the dynamic consistency property, we can consider each market participant's filtration to be generated by a set of LRBs, regardless of the time in which they enter the market, and without their views being inconsistent with other participants.

The dynamic consistency property was introduced in \citet{BHM1} with regard to Brownian random bridges, and was shown by the same authors to hold for gamma random bridges in \citep{BHM3}.

Fix a time $s<T$.
Given $L_{sT}$, we define a process $\{\eta_{t}\}$ by (\ref{eq:define_eta}).
We shall show that $\{\eta_t\}$ is an LRB.
At time $s$, the law of $\eta_{T}$ is
\begin{equation}
    \nu^*(A)=\nu_{s}(A+L_{sT}) \qquad \text{for all $A\in\mathcal{B}(\R)$,}
\end{equation}
where $A+y$ denotes the shifted set given by
\begin{equation}
    A+y= \left\{x: x-y\in A\right\}.
\end{equation}
Given the terminal value $\eta_{T}$, the finite-dimensional distributions of $\{\eta_{t}\}$ are given by
\begin{align}
    &\Q\left[\left.\eta_{s+t_1}\in\dd x_1,\ldots,\eta_{s+t_n}\in\dd x_n  \,\right| L_{sT}, \eta_{T}=z\right] \nonumber
    \\ &\qquad=\Q\left[\left.L_{s+t_1,T}-L_{sT}\in\dd x_1,\ldots,L_{s+t_n,T}-L_{sT}\in\dd x_n  \,\right| L_{sT}, L_{TT}-L_{sT}=z\right] \nonumber
    \\ &\qquad=\Q\left[\left.L_{s+t_1}-L_{s}\in\dd x_1,\ldots,L_{s+t_n}-L_{s}\in\dd x_n  \,\right| L_{s}, L_{T}-L_{s}=z\right] \nonumber
    \\ &\qquad=\Q\left[\left.L_{t_1}\in\dd x_1,\ldots,L_{t_n}\in\dd x_n  \,\right| L_{T-s}=z\right]      \nonumber
    \\ &\qquad=\frac{f_{T-s-t_n}\left(z-x_n\right)}{f_{T-s}\left(z\right)} \prod_{i=1}^{n} f_{t_i-t_{i-1}}\left(x_i-x_{i-1} \right),
\end{align}
for every $n\in\N_+$, every $0=t_0<t_1<\cdots<t_n<T-s$, and every $(x_1,\ldots,x_n)\in\R^n$, where $x_0=0$.
Then we have
\begin{multline}
\Q\left[\left.\eta_{s+t_1}\in\dd x_1,\ldots,\eta_{s+t_n}\in\dd x_n, \eta_{T}\in \dd z  \,\right| L_{sT}\right]
\\=\frac{f_{T-s-t_n}\left(z-x_n\right)}{f_{T-s}\left(z\right)} \prod_{i=1}^{n} f_{t_i-t_{i-1}}\left(x_i-x_{i-1} \right) \, \nu^*(\dd z).
\end{multline}
Comparison of this expression to (\ref{eq:LRBlaw}) shows that the process $\{\eta_{s+t}\}_{0\leq t\leq T-s}$ has the law $\lrbc([0,T-s],\{f_t\},\nu^*)$,
and so the law of $\{\eta_t\}_{s\leq t \leq T}$ is  $\lrbc([s,T],\{f_t\},\nu^*)$.

In the discrete case, we define $\{\eta_t\}$ by
\begin{align}
    \eta_t=M_{tT}-M_{sT} \qquad (s\leq t \leq T).
\end{align}
Then, given $M_{sT}$, $\{\eta_t\}$ has the law $\lrbd([s,T],\{Q_t\},P^*)$, where $P^*$ is defined by
\begin{equation}
    P^*(a)=P_s(a+M_{sT}).
\end{equation}

\subsection{Increments of LRBs}
The form of the transition law in Proposition \ref{prop:markov} shows that in general the increments of an LRB are not independent.
The special cases of LRBs with independent increments are discussed later.
A result that holds for all LRBs is that they have stationary increments:
\begin{prop}
    \label{prop:stat}
    For $s,t,u$ satisfying $0\leq s<u<T$ and $0<t\leq T-u$, we have
    \begin{align}
        \Q\left[L_{u+t, T}-L_{uT} \leq z \left|\, L_{sT} \right.\right]&=\Q[L_{s+t,T}-L_{sT}\leq z\left|\, L_{sT} \right.],
        \\ \intertext{and}
        \Q\left[M_{u+t, T}-M_{uT} \leq z \left|\, M_{sT} \right.\right]&=\Q[M_{s+t,T}-M_{sT}\leq z\left|\, M_{sT} \right.].
    \end{align}
\end{prop}
\begin{proof}
    We provide the proof for $\{L_{tT}\}$.
    The proof for $\{M_{tT}\}$ is similar.
    Throughout the proof we assume that $t<T-u$.
    The case $t=T-u$ follows from the stochastic continuity of $\{L_{tT}\}$.
    First we assume that $s=0$.
    From (\ref{eq:LRBtranslaw}), we have
    \begin{equation}
        \Q[L_{u+t, T}\in\dd y,L_{uT}\in\dd x]=\psi_{u+t}(\R;y)f_t(y-x)f_u(x) \d x \d y.
    \end{equation}
    Then we have
    \begin{align}
        \Q[L_{u+t, T}-L_{uT}\in\dd z,L_{uT}\in\dd x]
         &=\psi_{u+t}(\R;z+x)f_t(z)f_u(x) \d x \d z \nonumber
        \\ &=\int_{w=-\infty}^{\infty}\frac{f_{T-(u+t)}(w-z-x)}{f_T(w)} \d w \, f_t(z) f_u(x) \d x \d z.
    \end{align}
    Integrating over $x$ and changing the order of integration yields
    \begin{align}
        \Q[L_{u+t, T}-L_{uT}\in\dd z]
            &=\int_{w=-\infty}^{\infty}\int_{x=-\infty}^{\infty}f_{T-(u+t)}(w-z-x)f_u(x) \d x \, \frac{\dd w}{f_T(w)} \, f_t(z)  \d z \nonumber
        \\ &=\int_{w=-\infty}^{\infty}\frac{f_{T-t}(w-z)}{f_T(w)}\d w \, f_t(z)  \d z \nonumber
        \\ &=\psi_t(\R,z) f_t(z)  \d z \nonumber
        \\ &=\Q[L_{tT}\in\dd z].
    \end{align}

    For the case $s>0$, we use the dynamic consistency property.
    For $s$ fixed and $L_{sT}$ given, the process $\{\eta_{uT}\}_{s\leq u\leq T}=\{L_{uT}-L_{sT}\}_{s\leq u\leq T}$ is an LRB with the law $\lrbc([s,T],\{f_t\},\nu^*)$, where $\nu^*(A)=\nu_s(A+L_{sT})$.
    We have
    \begin{align}
        \Q\left[ L_{u+t , T} -L_{uT}\in \dd z \left|\, L_{sT} \right.\right]
         &=\Q\left[ \eta_{u+t , T} -\eta_{uT}\in \dd z \left|\, L_{sT} \right.\right] \nonumber
        \\ &=\Q\left[ \eta_{t T} \in \dd z \left|\, L_{sT} \right.\right] \nonumber
        \\ &=\int_{-\infty}^{\infty} \frac{f_{T-t}(w-z)}{f_{T-s}(w)}\nu^*(\dd w) \, f_{t-s}(z) \d z \nonumber
        \\ &=\int_{-\infty}^{\infty} \frac{f_{T-t}(w-z+L_{sT})}{f_{T-s}(w-L_{sT})}\nu_s(\dd w) \, f_{t-s}(z) \d z \nonumber
        \\ &=\frac{1}{\psi_s(\R;L_{sT})}
                \int_{-\infty}^{\infty} \frac{f_{T-t}(w-z+L_{sT})}{f_{T}(w)}\nu(\dd w) \, f_{t-s}(z) \d z \nonumber
        \\ &=\frac{\psi_t(\R;z+L_{sT})}{\psi_s(\R;L_{sT})}\, f_{t-s}(z) \d z \nonumber
        \\ &=\Q[L_{tT}-L_{sT} \in \dd z\left|\, L_{sT} \right].
    \end{align}
\end{proof}
When $\{L_{tT}\}$ is integrable, the stationary increments property offers enough structure to allow the calculation of the expected value of $L_{tT}$:
\begin{coro}
\label{coro:ExpectLRB}
If $\E[|L_{tT}|]<\infty$ for all $t\in(0,T]$ then
\begin{align}
        \E\left[L_{tT} \left|\, L_{sT} \right.\right]&=\frac{T-t}{T-s}L_{sT} + \frac{t-s}{T-s} \E\left[L_{TT} \left|\, L_{sT} \right.\right]
        &&(s<t),
        \\\intertext{and if $\E[|M_{tT}|]<\infty$ for all $t\in(0,T]$ then}
        \E\left[M_{tT} \left|\, M_{sT} \right.\right]&=\frac{T-t}{T-s}M_{sT} + \frac{t-s}{T-s} \E\left[M_{TT} \left|\, M_{sT} \right.\right]
        &&(s<t).
\end{align}
\end{coro}
\begin{proof}
    We provide the proof for $\{L_{tT}\}$.
    The proof for $\{M_{tT}\}$ is similar.
    The case $t=T$ is immediate, so we assume that $t<T$.
    First we consider the case $s=0$.
    Suppose that $t=mT/n$, where $m,n \in \mathbb{N}_+$ and $m<n$.
    We wish to show that
    \begin{equation}
        \E[L_{tT}]=\frac{m}{n} \E[L_{TT}].
    \end{equation}
    Writing $L(t,T)=L_{tT}$, define the random variables $\{\D{i}\}$ by
    \begin{equation}
        \D{i}=L\left(\tfrac{i}{n}T,T\right)-L\left(\tfrac{(i-1)}{n}T,T\right).
    \end{equation}
    It follows from Proposition \ref{prop:stat} that the $\D{i}$'s are identically distributed, and by assumption they are integrable.
    Hence we have
    \begin{equation}
            \E[\D{i}]=\frac{1}{n} \, \E\left[ \sum_{i=1}^n \D{i} \right]=\frac{1}{n}\,\E[L_{TT}].
    \end{equation}
    Then, as required, we have
    \begin{equation}
        \E\left[L\left(\tfrac{m}{n}T,T\right)\right]=\E\left[\sum_{i=1}^m \D{i}\right]=\frac{m}{n}\,\E[L_{TT}].
    \end{equation}

    For general $t$, choose an increasing sequence of positive rational numbers $\{q_i\}$ such that $\lim_{i\rightarrow\infty} q_i=t/T$.
    By use of the monotone convergence theorem one obtains
    \begin{equation}
        \E[L(t,T)]=\E\left[\lim_{i\rightarrow\infty} L\left(q_i T,T\right) \right]
            =\lim_{i\rightarrow\infty} \E\left[ L\left(q_i T,T\right) \right]=\frac{t}{T}\,\E[L_{TT}].
    \end{equation}

    For the case $s>0$, we use the dynamic consistency property.
    For $s$ fixed and $L_{sT}$ given, the process
    \begin{equation}
            \eta_{tT}=L_{tT}-L_{sT} \qquad (s\leq t\leq T)
    \end{equation}
    is an LRB with law $\lrbc([s,T],\{f_t\},\nu^*)$, where $\nu^*(A)=\nu_s(A+L_{sT})$.
    Then we have
    \begin{align}
        \E\left[L_{tT} \left|\, L_{sT} \right.\right]&=L_{sT}+\E[\eta_{tT}\left|\, L_{sT} \right.] \nonumber
        \\ &=L_{sT}+\frac{t-s}{T-s} \int_{-\infty}^{\infty} z \, \nu^*(\dd z) \nonumber
        \\ &=L_{sT}+\frac{t-s}{T-s} \int_{-\infty}^{\infty} (z-L_{sT}) \, \nu_s(\dd z) \nonumber
        \\ &=\frac{T-s}{T-s} L_{sT}+\frac{t-s}{T-s} \E\left[L_{TT} \left|\, L_{sT} \right.\right].
    \end{align}
\end{proof}

We have shown that the increments of LRBs are stationary, so it is natural to ask when the increments are independent,
i.e.~when is an LRB a \levy process?
The answer lies in the functional form of $\psi_t(\R;y)$.

For $0\leq s<t<T$, the likelihood that $L_{tT}=y$ given that $L_{sT}=x$ is
\begin{equation}
    \label{eq:qden}
    q(t,y;s,x)=\frac{\psi_t(\R;y)}{\psi_s(\R;x)}f_{t-s}(y-x).
\end{equation}
If $\{L_{tT}\}$ has stationary, independent increments then
\begin{equation}
    q(t,y;s,x)= q(t-s,y-x;0,0).
\end{equation}
Therefore the ratio
\begin{equation}
    \frac{\psi_t(\R;y)}{\psi_s(\R;x)}
\end{equation}
is a function of the differences $t-s$ and $y-x$.
Thus if we have
\begin{equation}
    \label{eq:psiind}
    \psi_t(\R;y)=a\exp(by+ct),
\end{equation}
for constants $a$, $b$ and $c$, then $\{L_{tT}\}$ is a \levy process.
There are constraints on $a$, $b$ and $c$ since (\ref{eq:qden}) is a probability density.
When $b=c=0$ we have $\nu(\dd z)=f_T(z)\d z$ which is the case where $\{L_{tT}\}\law\{L_t\}$.

\bigskip

\noindent{\bf Example.}
In the Brownian case we consider a process $\{W_{tT}\}$ with law
    \[\lrbc([0,T],\{f_t\},f_T(z-\th T) \d z),\]
where $f_t(x)$ is the normal density with zero mean and variance $t$, given by (\ref{eq:NormalDensity}).
In other words, $\{W_{tT}\}$ is a standard Brownian motion conditioned so that $W_{TT}$ is a normal random variable with mean $\th T$ and variance $T$.
In this case, we have
\begin{align}
    \psi_t(\R;y)&=\int_{-\infty}^{\infty}\frac{f_{T-t}(z-y)}{f_{T}(z)} f_T(z-\th T) \d z \nonumber
    \\ &=\exp\left(\th y-\frac{\th}{2}t\right).
\end{align}
Simplifying the expression for the transition densities of the process $\{W_{tT}\}$ allows one to verify that $\{W_{tT}\}$ is a Brownian motion with drift $\th$.
It is notable, by Girsanov's theorem, that $\{\psi_t(\R;W_t)\}$ is the Radon-Nikod\'ym density process that transforms a standard Brownian motion into a Brownian motion with drift $\th$.
Hence we can alternatively deduce that $\{W_{tT}\}$ is a Brownian motion with drift $\th$ from the analysis in Section \ref{sec:measure}.

\bigskip

\noindent{\bf Example.}
In the gamma case, we consider a process $\{\G_{tT}\}$ with law
\[\lrbc([0,T],\{f_t\},\kappa^{-1} f_T(z/\kappa) \d z),\]
where $f_t(x)$ is the gamma density with mean $mt$ and variance $mt$ defined by (\ref{eq:gamma_den_2}), and $\kappa>0$ is constant.
Then $\{\G_{tT}\}$ is a gamma process with mean $m$ and variance $m$ at $t=1$,
conditioned so that $\G_{TT}$ has a gamma distribution with mean $\kappa mT$ and variance $\kappa^2 m T$.
We have:
\begin{align}
    \psi_t(\R;y)&=\int_{-\infty}^{\infty}\frac{f_{T-t}(z-y)}{f_{T}(z)} \frac{f_T(z/\kappa)}{\kappa} \d z \nonumber
    \\ &=\kappa^{-mt}\exp\left((1-\kappa^{-1}) y\right).
\end{align}
The transition density of $\{\G_{tT}\}$ is
\begin{equation}
    \Q[\G_{tT}\in\dd y \,|\, \G_{sT}=x]=\1_{\{y>x\}}\frac{(y-x)^{m(t-s)-1} \e^{-(y-x)/\kappa}}{\kappa^{m(t-s)} \G(m(t-s))}\d y.
\end{equation}
Hence $\{\G_{tT}\}$ is a gamma process with mean $\kappa m$ and variance $\kappa^2m$ at $t=1$.

\subsubsection{Increment distributions}
Partition the time interval $[0,T]$ by $0=t_0<t_1<t_2<\cdots<t_n=T$.
Then define the increments $\{\D{i}\}_{i=1}^n$ and $\{\a_i\}_{i=1}^n$ by
\begin{align}
    \label{eq:deltadef}
    \D{i}&=L_{t_i,T}-L_{t_{i-1},T}
    \\ \a_i&=t_i-t_{i-1}.
\end{align}
Assume that $\nu$ has no \emph{continuous singular part} \citep{Sato1999}.
Denoting the Dirac delta function centred at $z$ by $\delta_z(x)$, $x\in\R$, we can write
\begin{equation}
    \label{eq:nu_nosing}
    \nu(\dd z)=\sum_{i=-\infty}^{\infty} v_i \delta_{z_i}(z)\d z + p(z) \d z,
\end{equation}
for some $\{a_i\}\subset \R$, $\{z_i\}\subset \R_+$, and $p:\R\rightarrow\R_+$.
Here $p(z)$ is the density of the continuous part of $\nu$, and $v_i$ is a point mass of $\nu$ located at $z_i$.
By (\ref{eq:LRBlaw}), the joint law of the random vector $(\D{1},\ldots , \D{n})^{\tp}$ is given by
\begin{equation}
    \label{eq:incden}
    \Q[\D{1}\in\dd y_1\ldots,\D{n}\in\dd y_n]= \widetilde{f}\left( \sum_{i=1}^ny_i \right)
                        \prod_{i=1}^n f_{\a_i}(y_i) \d y_i,
\end{equation}
where
\begin{equation}
    \widetilde{f}(z)=\frac{p(z)+\sum_{i=-\infty}^{\infty} v_i \delta_{z_i}(z)}{f_T(z)}.
\end{equation}
Equation (\ref{eq:incden}) shows that $(\D{1},\ldots , \D{n})^{\tp}$ has a generalized multivariate Liouville distribution as defined by \citet{GRIV}.
The classical multivariate Liouville distribution is obtained when $f_t(x)$ is the density of a gamma distribution (see \citep{GRI,GRII,GRIII,FKN1990}).
A survey of Liouville distributions can be found in \citet{GR2001}.
\citet{BNJ1991} construct a generalized Liouville distribution by conditioning a vector of independent inverse-Gaussian random variables on their sum.

In the discrete case, the joint distribution of increments also has a generalized Liouville distribution.
Define the increments $\{D_i\}$ by
\begin{equation}
    D_{i}=M_{t_i,T}-M_{t_{i-1},T}.
\end{equation}
Then we can write
\begin{equation}
    \Q[D_{1}\in\dd y_1\ldots,D_{n}\in\dd y_n]= \widetilde{Q}\left( \sum_{i=1}^ny_i \right)
                        \prod_{i=1}^n \dd Q_{\a_i}(y_i),
\end{equation}
where
\begin{equation}
    \widetilde{Q}(z)=\frac{\sum_{i=-\infty}^{\infty} P(a_i) \delta_{a_i}(z)}{Q_T(z)}.
\end{equation}

\subsubsection{The reordering of increments}
We are able to extend the Markov property of LRBs.
If we partition the path of an LRB into increments, then the Markov property means that future increments depend on the past only through the \emph{sum} of past increments.
We shall show that for LRBs the ordering of the increments does not matter for this to hold---%
given the values of any set of increments of an LRB (past or future), the other increments depend on this subset only through the sum of its elements.

Let $\pi$ be a permutation of $\{1,2,\ldots,n\}$.
We define the partial sum $S^{\pi}_m$ by
\begin{equation}
    S^{\pi}_m=\sum_{i=1}^m \D{\pi(i)} \qquad \text{for $m=1,2,\ldots,n$,}
\end{equation}
where the $\{\D{i}\}$ are defined as in (\ref{eq:deltadef});
and we define the partition $0=t^{\pi}_0<t^{\pi}_1<\cdots<t^{\pi}_n=T$ by
\begin{equation}
    t^{\pi}_{j+1}=\sum_{i=1}^{j}\a_{\pi(i)} \qquad \text{for $j=1,2,\ldots,n-1$.}
\end{equation}

\begin{prop} \label{prop:inc}
    We may extend the Markov property of $\{L_{tT}\}$ to the following:
    \begin{multline}
        \label{eq:increment_prop}
        \Q\left[\D{\pi(m+1)}\leq y_{m+1},\ldots,\D{\pi(n)}\leq y_n \left|\, \D{\pi(1)},\ldots,\D{\pi(m)}\right.\right]=
            \\ \Q\left[\left.\D{\pi(m+1)}\leq y_{m+1},\ldots,\D{\pi(n)}\leq y_{n} \,\right| S^{\pi}_m \right].
    \end{multline}
    If $\nu$ has no singular continuous part, then
    \begin{multline}
        \Q\left[\left.\D{\pi(m+1)}\in\dd y_{m+1},\ldots,\D{\pi(n)}\in\dd y_n \,\right| S^{\pi}_m \right]=
         \\ \frac{\widetilde{f}\left(S^{\pi}_m + \sum_{i=m+1}^n y_i\right)}{\psi_{t^{\pi}_m}(\R;S^{\pi}_m)}
        \prod_{i=m+1}^n f_{\a_{\pi(i)}}(y_i)\d y_i.
    \end{multline}
\end{prop}
\begin{proof}
    Define the increments $\{\tD{i}\}$ by
    \begin{equation}
        \tD{i}=L_{t^{\pi}_{n},T}-L_{t^{\pi}_{n-1},T}.
    \end{equation}
    The law of the random vector $(\tD{1},\ldots , \tD{n-1},\sum_{1}^n\tD{i})^{\tp}$ is given by
    \begin{multline}
        \Q\left[\tD{1}\in\dd y_{1},\ldots,\tD{n-1}\in \dd y_{n-1},\sum_{i=1}^n\tD{i}\in\dd z\right]=
            \\  \frac{\nu(\dd z)}{f_T(z)}f_{\a_{\pi(n)}}\left(z-\sum_{i=1}^{n-1}y_i \right)\prod_{i=1}^{n-1} f_{\a_{\pi(i)}}(y_i) \d y_i.
    \end{multline}
    This is also the law of $(\D{\pi(1)},\ldots,\D{\pi(n-1)},\sum_1^n\D{\pi(i)})^{\tp}$; hence
    \begin{equation}
        (\D{\pi(1)},\ldots,\D{\pi(n)})\law (\tD{1},\ldots , \tD{n}).
    \end{equation}
    The Markov property of LRBs gives
    \begin{multline}
        \Q\left[\tD{m+1}\leq y_{m+1},\ldots,\tD{n}\leq y_n \left|\, \tD{1},\ldots,\tD{m}\right.\right]=
            \\ \Q\left[\tD{m+1}\leq y_{m+1},\ldots,\tD{n}\leq y_{n} \left|\, \sum_{i=1}^m\tD{i} \right.\right],
    \end{multline}
    and so we have
    \begin{multline}
        \Q\left[\D{\pi(m+1)}\leq y_{m+1},\ldots,\D{\pi(n)}\leq y_n \left|\, \D{\pi(1)},\ldots,\D{\pi(m)}\right.\right]=
            \\ \Q\left[\left.\D{\pi(m+1)}\leq y_{m+1},\ldots,\D{\pi(n)}\leq y_{n} \,\right| S^{\pi}_m \right].
    \end{multline}
    This proves the first part of the proposition.

    For the second part of the proof we assume that $\nu$ takes the form (\ref{eq:nu_nosing}).
    Note that
    \begin{equation}
        L_{t^{\pi}_m,T}=\sum_{i=1}^m \tD{i},
    \end{equation}
    and that the density of $L_{t^{\pi}_m,T}$ is
    \begin{equation}
        x \mapsto f_{t^{\pi}_m}(x)\psi_{t^{\pi}_m}(\R;x)
        =\int_{z=-\infty}^{\infty} \frac{f_{t^{\pi}_m}(x)f_{T-t^{\pi}_m}(z-x)}{f_T(z)}\,\nu(\dd z).
    \end{equation}
  The elements of the vector $(L_{t^{\pi}_m,T},\tD{m+1},\ldots,\tD{n})^{\tp}$ are non-overlapping increments of $\{L_{tT}\}$, and the law of the vector is given by
    \begin{multline}
        \Q\left[L_{t^{\pi}_m,T}\in \dd x, \tD{m+1}\in\dd y_{m+1},\ldots,\tD{n}\in \dd y_{n}\right]=
        \\\widetilde{f}\left(x+ \sum_{i=m+1}^ny_i \right) f_{t^{\pi}_m}(x) \d x \,
                    \prod_{i=m+1}^n f_{\a_{\pi(i)}}(y_i) \d y_i.
    \end{multline}
    Thus we have
    \begin{align}
        &\Q\left[\left.\tD{m+1}\in \dd y_{m+1},\ldots,\tD{n}\in \dd y_n \, \right| L_{t^{\pi}_m,T}=x\right] \nonumber
        \\ &\qquad\qquad\qquad=\frac{\Q\left[\tD{m+1}\in \dd y_{m+1},\ldots,\tD{n}\in \dd y_n , L_{t^{\pi}_m,T}\in \dd x \right]}
                                {\Q\left[L_{t^{\pi}_m,T}\in \dd x \right]} \nonumber
        \\ &\qquad\qquad\qquad=\frac{\widetilde{f}\left(x+ \sum_{i=m+1}^ny_i \right)\prod_{i=m+1}^n f_{\a_{\pi(i)}}(y_i)}
                    {\psi_{t^{\pi}_m}(\R;S^{\pi}_m)}.
    \end{align}
\end{proof}

We note that \citet{GRIV} prove that if $(\D{1},\D{2},\ldots,\D{n})^{\tp}$ has a generalized Liouville distribution then equation (\ref{eq:increment_prop}) holds.

\bigskip
We can use Proposition \ref{prop:inc} to extend the dynamic consistency property.
In particular we have the following:

\begin{coro}\label{coro:inc}
    \begin{enumerate}
        \item[a.]
            Fix times $s_1,T_1$ satisfying $0<T_1\leq T-s_1$.
            The time-shifted, space-shifted partial process
            \begin{equation}
                \eta_{t,T_1}^{(1)}=L_{s_1+t,T}-L_{s_1,T}, \qquad (0\leq t \leq T_1),
            \end{equation}
            is an LRB with the law $\lrbc([0,T_1],\{f_t\},\nu^{(1)})$, where $\nu^{(1)}$ is a probability law on $\R$ with density $f_{T_1}(x)\psi_{T_1}(\R;x)$.
        \item[b.]
            Construct the partial processes $\{\eta^{(i)}_{t,T_i}\}$, $i=1,\ldots,n$,
            from non-overlapping portions of $\{L_{tT}\}$ in a similar way to that above.
            The intervals $[s_i,s_i+T_i]$, $i=1,\ldots,n$, are non-overlapping except possibly at the endpoints.
            Set $\eta^{(i)}_{t,T_i}=\eta^{(i)}_{T_i,T_i}$ when $t>T_i$.
            If $u>t$, then
            \begin{multline}
                \Q\left[\left. \eta^{(1)}_{u,T_1}-\eta^{(1)}_{t,T_1}\leq x_1,\ldots,\eta^{(n)}_{u,T_n}-\eta^{(n)}_{t,T_n}\leq x_n \,\right| \F^{\eta}_t \right]=
                \\\Q\left[ \eta^{(1)}_{u,T_1}-\eta^{(1)}_{t,T_1}\leq x_1,\ldots,\eta^{(n)}_{u,T_n}-\eta^{(n)}_{t,T_n}\leq x_n \left|\, \sum_{i=1}^n \eta^{(i)}_{t,T_i} \right.\right],
            \end{multline}
            where
            \begin{equation}
                \F^{\eta}_{t}=\s\left(\left\{ \eta^{(i)}_{s,T_i} \right\}_{0\leq s \leq t}, i=1,2,\ldots,n  \right).
            \end{equation}
    \end{enumerate}
\end{coro}

\begin{rem}
The partial processes of Corollary \ref{coro:inc} are dependent, and
\begin{equation}
    \Q\left[\left. \eta^{(i)}_{tT} \in \dd x \,\right| \F^{\eta}_s \right]=
    \Q\left[\eta^{(i)}_{tT} \in \dd x \left|\, \eta^{(i)}_{sT}, \sum_{j=1}^{n} \eta^{(j)}_{sT} \right.\right],
\end{equation}
for $0\leq s<t \leq T$.
\end{rem}

We state but do not prove a discrete analogue of Proposition \ref{prop:inc}, which is as follows:
\begin{prop}
    One can extend the Markov property of $\{M_{tT}\}$ to the following:
    \begin{multline}
        \Q\left[D_{\pi(m+1)}\leq y_{m+1},\ldots,D_{\pi(n)}\leq y_n \left|\, D_{\pi(1)},\ldots,D_{\pi(m)}\right.\right]=
            \\ \Q\left[\left.D_{\pi(m+1)}\leq y_{m+1},\ldots,D_{\pi(n)}\leq y_{n} \,\right| R^{\pi}_m \right],
    \end{multline}
    where $R^{\pi}_m=\sum_{i=1}^m D_{\pi(i)}$.
    Furthermore,
    \begin{equation}
        \Q\left[\left.D_{\pi(m+1)}= y_{m+1},\ldots,D_{\pi(n)}= y_n \,\right| D^{\pi}_m \right]=
          \frac{\widetilde{Q}\left(R^{\pi}_m + \sum_{i=m+1}^n y_i\right)}{\sum_{k=-\infty}^{\infty}\phi_{t^{\pi}_m}(a_k;R^{\pi}_m)}
        \prod_{i=m+1}^n Q_{\a_{\pi(i)}}(y_i).
    \end{equation}
\end{prop}

Corollary \ref{coro:inc} can be extended to include LRBs with discrete state-spaces.

%
%

\section{Information-based asset pricing}

\subsection{BHM framework}
We begin with a brief overview of the BHM framework.
The approach was applied to credit risk in \citet{BHM1}, and this was extended to include stochastic interest rates in \citet{YR2007}.
A general asset pricing framework was proposed in \citet{BHM2} (see also \citet{MPhD2006}), and there have also been applications to inflation modelling (\citet{HM2008}), insider trading (\citet{BDFH2008}), insurance (\citet{BHM3}), and interest rate theory (\citet{HM2}).

We fix a finite time horizon $[0,T]$ and a probability space $(\Omega,\F,\Q)$.
We assume that the risk-free rate of interest $\{r_t\}$ is deterministic, and that
$r_t>0$ and $\int_t^{\infty}r_u \d u=\infty$, for all $t>0$.
Then the time-$s$ (no-arbitrage) price of a risk-free, zero-coupon bond maturing at time $t$ (paying a nominal amount of unity) is
\begin{equation}
    P_{st}=\exp\left( -\int_s^t r_u \d u \right) \qquad (s\leq t).
\end{equation}
For $t<T$, the time-$t$ price of a contingent cash flow $H_T$, due at time $T$, is given by an expression of the form
\begin{equation}
    \label{eq:price_formula}
    H_{tT}=P_{tT} \, \E[H_T \,|\, \F_t],
\end{equation}
where $\{\F_t\}$ is the \emph{market filtration}.
The sigma-algebra $\F_t$ represents the information available to market participants at time $t$.
In order for equation (\ref{eq:price_formula}) to be consistent with the theory of no-arbitrage pricing, we interpret $\Q$ to be the risk-neutral measure.

In such a set-up, the dynamics of the price process $\{H_{tT}\}$ are implicitly determined by the evolution of the market filtration $\{\F_t\}$.
We assume the existence of a (possibly multi-dimensional) \emph{information process} $\{\xi_{tT}\}_{0\leq t \leq T}$ such that
\begin{equation}
    \F_t=\s \left(\{\xi_{sT}\}_{0\leq s \leq t} \right).
\end{equation}
Thus $\{\xi_{tT}\}$ is responsible for the delivery of all information to the market participants.
The task of modelling the emergence of information in the market is reduced to that of specifying the law of the information process $\{\xi_{tT}\}$.

\subsubsection{Single $X$-factor market}
We assume that the cash flow $H_T$ can be written in the form
\begin{equation}
    H_T=h(X_T),
\end{equation}
for some function $h(x)$, and some market factor $X_T$.
We call $X_T$ an $X$-factor.
We assume that $\{\xi_{tT}\}$ is a one-dimensional process such that $\xi_{TT}=X_T$.
Then we have
\begin{equation}
    H_{tT}=P_{tT}\, \E[h(X_T) \,|\, \F_t]=P_{tT}\, \E[h(\xi_{TT}) \,|\, \F_t],
\end{equation}
which ensures that $H_{TT}=H_T$.
In the case where $\{\xi_{tT}\}$ is a Markov process, we have
\begin{equation}
    H_{tT}= P_{tT} \, \E[h(\xi_{TT}) \,|\, \xi_{tT}].
\end{equation}

\subsubsection{Multiple $X$-factor market}
In the more general framework, we model an asset that generates $N$ cash flows $H_{T_1},H_{T_2},\ldots,H_{T_N}$, which are to be received on the dates $T_1\leq T_2\leq \cdots \leq T_N$, respectively.
At time $T_k$, we assume that the vector of $X$-factors $X_{T_k}\in\R^{n_k}$ ($n_k\in\p{N}_+$) is revealed to the market, and we write
\begin{equation}
    X_{T_k}=\left(X^{(1)}_{T_k},X^{(2)}_{T_k},\ldots,X^{(n_k)}_{T_k}\right)^{\tp}.
\end{equation}
We assume the $X$-factors are mutually independent, and that
\begin{equation}
    H_{T_k}=h_k(X_{T_1},X_{T_2},\ldots,X_{T_k}),
\end{equation}
for some $h_k:\R^{n_1}\times\R^{n_2}\times\cdots\times\R^{n_k}\rightarrow\R$ which we call a cash-flow function.
For each $X$-factor $X^{(i)}_{T_j}$, there is a factor information process $\{\xi_t^{(i,j)}\}$ such that $\xi_t^{(i,j)}=X^{(i)}_{T_j}$ for $t\geq T_j$, and the factor information processes are mutually independent.
Setting $T=T_N$, we define the market information process $\{\xi_{tT}\}$ to be an $\R^{n_1+n_2+\cdots+n_N}$-valued process with each of its elements being a factor information process.
The market filtration $\{\F_t\}$ is generated by $\{\xi_{tT}\}$.
By construction, $H_{T_k}$ is $\F_t$-measurable for $t\geq T_k$.
The time-$t$ price of the cash flow $H_{T_k}$ is
\begin{equation}
    H_{tT}^{(k)}=\left\{
        \begin{aligned}
            &P_{t,T_k} \, \E\left[\left. h_k(X_{T_1},X_{T_2},\ldots,X_{T_k}) \,\right| \F_t\right] && \text{for $t< T_k$,}
            \\ &0 &&\text{for $t\geq T_k$.}
        \end{aligned}
        \right.
\end{equation}
Here we adopt the convention that cash flows have nil value at the time that they are due.
In other words, prices are quoted on an ex-dividend basis.
In this way the process $\{H^{(k)}_t\}$ is right-continuous at $t=T_k$.
The asset price process is then
\begin{equation}
    H_{tT}= \sum_{k=1}^n  H^{(k)}_{tT} \qquad (0\leq t\leq T).
\end{equation}

\subsection{\levy bridge information}
We consider a market with a single factor, which we denote $X_T$.
This $X$-factor is the size of a contingent cash flow to be received at time $T>0$, so we take $h(x)=x$.
For example, $X_T$ could be the redemption amount of a credit risky bond.
$X_T$ is assumed to be integrable and to have the \emph{a priori} probability law $\nu$ (we exclude the case where $X_T$ is constant).
Information is supplied to the market by an information process $\{\xi_{tT}\}$.
The law of $\{\xi_{tT}\}$ is $\lrbc([0,T],\{f_t\},\nu)$, and we set $\xi_{TT}=X_T$.
We assume throughout this section that the information process has a continuous state-space; the results can be extended to include LRB information processes with discrete state-spaces.

Since the information process has the Markov property, the price of the cash flow $X_T$ is given by
\begin{equation}
    X_{tT}=P_{tT}\, \E\left[X_T\left|\, \xi_{tT} \right. \right] \qquad (0\leq t \leq T).
\end{equation}
We note that $X_T$ is $\F_T$-measurable and $X_{TT}=X_T$, but $X_T$ is not $\F_t$-measurable for $t<T$ since we have excluded the case where $X_T$ is constant.
For $t\in(0,T)$, the $\F_t$-conditional law of $X_T$ as given by equation (\ref{eq:C_terminal}) is
\begin{equation}
    \nu_{t}(\dd z)=\frac{\psi_t(\dd z;\xi_{tT})}{\psi_t(\R;\xi_{tT})},
\end{equation}
where
\begin{equation}
    \psi_t(\dd z;\xi)=\frac{f_{T-t}(z-\xi)}{f_T(z)} \d z.
\end{equation}
Then we have
\begin{equation}
    X_{tT}=P_{tT} \int_{-\infty}^{\infty} z \, \nu_t(\dd z).
\end{equation}
When $\nu$ admits a density $p(z)$, the $\F_t$-conditional density of $X_T$ exists and is given by
\begin{equation}
    p_t(z)= \frac{f_{T-t}(z-\xi_{tT})p(z)}{\psi_t(\R;\xi_{tT}) f_T(z)}.
\end{equation}

\bigskip

\noindent{\bf Example.}
In the Brownian case the price is
\begin{equation}
    X_{tT}=P_{tT} \frac{\int_{-\infty}^{\infty} z \, \e^{\frac{1}{T-t}\left[\xi_{tT}z-\half \frac{t}{T}z^2\right]}\,\nu(\dd z)}
                {\int_{-\infty}^{\infty} \e^{\frac{1}{T-t}\left[\xi_{tT}z-\half \frac{t}{T}z^2\right]}\,\nu(\dd z)}.
\end{equation}
The following SDE can be derived for $\{X_{tT}\}$ (see \citep{BHM1,BHM2,MPhD2006,YR2007}):
\begin{equation}
    \label{eq:SDE}
    \dd X_{tT}=r_t X_{tT} \d t+ \frac{P_{tT} \var[X_T \,|\,\xi_{tT}]}{T-t} \d W_t,
\end{equation}
where $\{W_t\}$ is an $\{\F_t\}$-Brownian motion.

\bigskip

\noindent{\bf Example.}
In the gamma case we have
\begin{equation}
    X_{tT}=P_{tT}\frac{\int_{\xi_{tT}}^{\infty}(z-\xi_{tT})^{m(T-t)-1}z^{2-mT}\,\nu(\dd z)}{\int_{\xi_{tT}}^{\infty}(z-\xi_{tT})^{m(T-t)-1}z^{1-mT}\,\nu(\dd z)}.
\end{equation}

\subsection{European option pricing} \label{sec:Call_Option}
We consider the problem of pricing a European option on the price $X_{tT}$ at time $t$.
For a strike price $K$ and $0\leq s<t<T$, the time-$s$ price of a $t$-maturity call option on $X_{tT}$ is
\begin{equation}
    C_{st}=P_{st} \, \E\left[\left.(X_{tT}-K)^+\,\right| \xi_{sT} \right].
\end{equation}
The expectation can be expanded in the form
\begin{align}
    \E_{\Q}\left[\left.(X_{tT}-K)^+\,\right| \xi_{sT}\right]
     &=\E_{\Q}\left[\left.\left(P_{tT}\,\E_{\Q}[X_T\,|\,\xi_{tT}]-K  \right)^+\,\right| \xi_{sT}\right] \nonumber
    \\ &=\E_{\Q}\left[\left.\left(\int_{-\infty}^{\infty}(P_{tT}z-K) \, \nu_t(\dd z) \right)^+\,\right| \xi_{sT}\right] \nonumber
    \\ &=\E_{\Q}\left[\left.\frac{1}{\psi_t(\R;\xi_{tT})} \left(\int_{-\infty}^{\infty}(P_{tT}z-K) \, \psi_t(\dd z;\xi_{tT}) \right)^+\,\right| \xi_{sT}\right].
\end{align}
Recall that the Radon-Nikodym density process
\begin{equation}
    \left.\frac{\dd \p{L}}{\dd \Q}\right|_{\F_t}=\psi_t(\R;\xi_{tT})^{-1}
\end{equation}
defines a measure $\p{L}$ under which $\{\xi_{tT}\}_{0\leq t<T}$ is a \levy process.
By changing measure, we find that the expectation is
\begin{multline}
    \frac{1}{\psi_s(\R;\xi_{sT})}\,
    \E_{\p{L}}\left[\left.\left(\int_{-\infty}^{\infty}(P_{tT}z-K) \, \psi_t(\dd z;\xi_{tT}) \right)^+\,\right| \xi_{sT}\right]=
    \\ \frac{1}{\psi_s(\R;\xi_{sT})}\,
    \E_{\p{L}}\left[\left.\left(\int_{-\infty}^{\infty}(P_{tT}z-K)\frac{f_{T-t}(z-\xi_{tT})}{f_T(z)} \, \nu(\dd z) \right)^+\,\right| \xi_{sT}\right].
\end{multline}
Equation (\ref{eq:density_conditions}) states that $0<f_{T-s}(z-\xi_{sT})<\infty$.
Thus we can write the expectation in terms of the $\xi_{sT}$-conditional terminal law $\nu_s$ in the form
\begin{multline}
    \E_{\p{L}}\left[\left.\left(\int_{-\infty}^{\infty}
            (P_{tT}z-K)\frac{f_{T-t}(z-\xi_{tT})}{f_{T-s}(z-\xi_{sT})} \, \nu_s(\dd z) \right)^+\,\right| \xi_{sT}\right]
    =\\ \int_{-\infty}^{\infty} \left(\int_{-\infty}^{\infty}
            (P_{tT}z-K)\frac{f_{T-t}(z-x)}{f_{T-s}(z-\xi_{sT})} \, \nu_s(\dd z) \right)^+f_{t-s}(x-\xi_{sT}) \d x.
\end{multline}
We defined the (marginal) \levy bridge density $f_{tT}(x;z)$ by
\begin{equation}
    f_{tT}(x;z)=\frac{f_{T-t}(z-x)f_{t}(x)}{f_{T}(z)}.
\end{equation}
From this we can define the $\xi_{sT}$-dependent law $\mu_{st}(\dd x;z)$ by
\begin{equation}
    \mu_{st}(\dd x;z)=f_{t-s,T-s}(x-\xi_{sT},z-\xi_{sT}) \d x.
\end{equation}
Thus $\mu_{st}(\dd x;z)$ is the time-$t$ marginal law of a \levy bridge starting at the value $\xi_{sT}$ at time $s$, and terminating at the value $z$ at time $T$.
Defining the set $B_t$ by
\begin{equation}
    B_{t}=\left\{ x\in\R: \int_{-\infty}^{\infty}(P_{tT}z-K)\frac{f_{T-t}(z-x)}{f_{T}(z)} \, \nu(\dd z)>0 \right\},
\end{equation}
the expectation reduces to
\begin{equation}
    \int_{-\infty}^{\infty}(P_{tT}z-K)\mu_{st}(B_{t};z) \, \nu_s(\dd z).
\end{equation}
The option price is then given by
\begin{equation}
    C_{st}=P_{st} \int_{-\infty}^{\infty}(P_{tT}z-K)\mu_{st}(B_{t};z) \, \nu_s(\dd z).
\end{equation}

We can write $X_{tT}=\Lambda(t,\xi_{tT})$, for $\Lambda$ a deterministic function.
The set $B_t$ can then be written
\begin{equation}
    B_t=\left\{ \xi\in\R: \Lambda(t,\xi)>K \right\}.
\end{equation}
We see that if $\Lambda$ is increasing in its second argument then $B_t=(\xi^*_t,\infty)$ for some critical value $\xi^*_t$ of the information process.
$\Lambda$ is monotonic if $\{\xi_{tT}\}$ is a \levy process.

\bigskip

\noindent{\bf Example.}
In the Brownian case we have
\begin{equation}
    \Lambda(t,x)=P_{tT} \frac{\int_{-\infty}^{\infty} z \, \e^{\frac{1}{T-t}\left[x z-\half \frac{t}{T}z^2\right]}\,\nu(\dd z)}
                {\int_{-\infty}^{\infty} \e^{\frac{1}{T-t}\left[x z-\half \frac{t}{T}z^2\right]}\,\nu(\dd z)}.
\end{equation}
It can be shown that the function $\Lambda$ is increasing in its second argument (see \citep{BHM2,YR2007}); hence $B_t=(\xi_t^*,\infty)$ for the unique $\xi_t^*$ satisfying $\Lambda(t,\xi_t^*)=K$.
A short calculation verifies that $\mu_{st}(\dd x;z)$ is the normal law with mean $M(z)$ and variance $V$ given by
\begin{align}
    &M(z)=\frac{T-t}{T-s} \xi_{sT} + \frac{t-s}{T-s} z,
    &&V= \frac{t-s}{T-s}(T-t).
\end{align}
This is the time-$t$ marginal law of a Brownian bridge starting from the value $\xi_{sT}$ at time $s$, and finishing at the value $z$ at time $T$.
We have
\begin{equation}
    \mu_{st}(B_t;z)=1-\Phi\left[\frac{\xi_t^*-M(z)}{\sqrt{V}}\right]=\Phi\left[\frac{M(z)-\xi_t^*}{\sqrt{V}}  \right],
\end{equation}
where $\Phi[x]$ is the standard normal distribution function.
The option price is then
\begin{equation}
    C_{st}=P_{sT} \int_{-\infty}^{\infty} z\, \Phi\left[\frac{M(z)-\xi_t^*}{\sqrt{V}}\right] \nu_s(\dd z)
            +P_{st}K \int_{-\infty}^{\infty}\Phi\left[\frac{M(z)-\xi_t^*}{\sqrt{V}}  \right] \nu_s(\dd z).
\end{equation}

\bigskip

\noindent{\bf Example.}
In the gamma case we have
\begin{equation}
    \Lambda(t,x)=P_{tT}\frac{\int_{x}^{\infty}(z-x)^{m(T-t)-1}z^{2-mT}\,\nu(\dd z)}
                {\int_{x}^{\infty}(z-x)^{m(T-t)-1}z^{1-mT}\,\nu(\dd z)}.
\end{equation}
The monotonicity of $\Lambda(t,x)$ in $x$ was proved for $m(T-t)>1$ by \citet{BHM3}.
The authors also give a numerical example where $\Lambda(t,x)$ was not monotonic in $x$ for $m(T-t)<1$.
For all $t\in(0,T)$, we have
\begin{equation}
    \mu_{st}(\dd x;z)=\1_{\{\xi_{sT}<x< z\}}\, k(z)^{-1}
        \left(\frac{x-\xi_{st}}{z-\xi_{sT}}\right)^{m(t-s)-1} \left(\frac{z-y}{z-\xi_{sT}}\right)^{m(T-t)-1}
                     \d x,
\end{equation}
where $k(z)$ is the normalising constant
\begin{equation}
    k(z)={(z-\xi_{sT}) \, \mathrm{B}[m(t-s),m(T-t)]}.
\end{equation}
Hence $\mu_{st}(\dd x;z)$ is an $(z-\xi_{sT})$-scaled, $\xi_{sT}$-shifted, beta law with parameters $\a=m(t-s)$ and $\b=m(T-t)$.
This is the time-$t$ marginal law of a gamma bridge starting at the value $\xi_{sT}$ at time $s$, and terminating at the value $x$ at time $T$.
When $m(T-t)>1$, a critical $\xi^*_t$ exists such that $\Lambda(t,\xi_t^*)=K$.
Then $B_t=(\xi^*_t,\infty)$, and
\begin{align}
    \mu_{st}(B_t;z)&=1-I\left[\frac{\xi^*_t-\xi_{sT}}{z-\xi_{sT}};m(t-s),m(T-t)\right] \nonumber
    \\ &=I\left[\frac{z-\xi^*_t}{z-\xi_{sT}};m(T-t),m(t-s)\right].
\end{align}
Here $I[z;\a,\b]$ is the regularized incomplete beta function, defined for $\a,\b>0$ by
\begin{equation}
    I[z;\a,\b]=\frac{1}{\mathrm{B}[\a,\b]} {\int_0^z x^{\a-1} (1-x)^{\b-1} \d x}.
\end{equation}
The option price is then given by
\begin{multline}
    C_{st}=P_{sT} \int_{\xi_{sT}}^{\infty} z\, I\left[\frac{z-\xi^*_t}{z-\xi_{sT}};m(T-t),m(t-s)\right] \nu_s(\dd z)
            \\+P_{st}K \int_{\xi_{sT}}^{\infty}I\left[\frac{z-\xi^*_t}{z-\xi_{sT}};m(T-t),m(t-s)\right] \nu_s(\dd z).
\end{multline}

\subsection{Binary bond} \label{sec:Binary_Bond}
The simplest non-trivial contingent cash flow is $X_T\in\{k_0,k_1\}$, for $k_0< k_1$.
This is the pay-off from a zero-coupon, credit-risky bond that has principle $k_1$, and a fixed recovery rate $k_0/k_1$ on default.
Assume that, \emph{a priori}, $\Q[X_T=k_0]=p>0$ and $\Q[X_T=k_1]=1-p$.
Then
\begin{align}
    \Q[X_{T}=k_0 \,|\, \xi_{tT} ]&=\left(1+{\frac{f_T(k_0)}{f_T(k_1)}\frac{f_{T-t}(k_1-\xi_{tT})}{f_{T-t}(k_0-\xi_{tT})}\frac{1-p}{p}}  \right)^{-1},
    \\\intertext{and}
    \Q[X_{T}=k_1 \,|\, \xi_{tT} ]&=\left(1+{\frac{f_T(k_1)}{f_T(k_0)}\frac{f_{T-t}(k_0-\xi_{tT})}{f_{T-t}(k_1-\xi_{tT})}\frac{p}{1-p}}  \right)^{-1}.
\end{align}
The bond price process $\{X_{tT}\}$ associated with the given terminal cash flow is given by
\begin{equation}
    X_{tT}=P_{tT} \left(k_0\,\Q[X_{T}=k_0 \,|\, \xi_{tT} ]+k_1\,\Q[X_{T}=k_1 \,|\, \xi_{tT} ]  \right) \qquad(0\leq t \leq T).
\end{equation}

\bigskip
\noindent{\bf Example.}
In the Brownian case we have
\begin{align}
    \Q[X_{T}=k_0 \,|\, \xi_{tT} ]&=
        \left(1+\exp\left[-\half \frac{k_1-k_0}{T-t}(\tfrac{t}{T}(k_0+k_1)-2\xi_{tT})\right] \frac{1-p}{p}  \right)^{-1},
    \\\intertext{and}
    \Q[X_{T}=k_1 \,|\, \xi_{tT} ]&=
     \left(1+\exp\left[\half \frac{k_1-k_0}{T-t}(\tfrac{t}{T}(k_0+k_1)-2\xi_{tT})\right] \frac{p}{1-p}  \right)^{-1}.
\end{align}
Writing $\rho_i=\Q[X_{T}=k_i \,|\, \xi_{tT} ]$, note that
\begin{align}
    \var[X_T\,|\,\xi_{tT}]&=(k_1-k_0)^2\rho_1 \rho_0 \nonumber
    \\&=-(k_0-k_0\rho_0-k_1\rho_1)(k_1-k_0\rho_0-k_1\rho_1) \nonumber
    \\ &=-(k_0-X_{tT})(k_1-X_{tT}).
\end{align}
Thus, recalling (\ref{eq:SDE}), we see that the SDE of $\{X_{tT}\}$ is
\begin{equation}
    \dd X_{tT}=r_t X_{tT} \d t-\frac{P_{tT}(k_0-X_{tT})(k_1-X_{tT})}{T-t} \d W_t,
\end{equation}
with the initial condition $X_{0T}=k_0p+k_1(1-p)$.
For $K\in(P_{tT}k_0,P_{tT}k_1)$, we are able to solve the equation $\Lambda(t,x)=K$ for $x$.
We have
\begin{align}
    \Lambda(t,x)&=P_{tT} \left(k_0\, \Q[X_T=k_0\,|\,\xi_{tT}=x]+k_1\, \Q[X_T=k_1\,|\,\xi_{tT}=x]\right) \nonumber
        \\ &=P_{tT} \left(k_1-(k_1-k_0)\, \Q[X_T=k_0\,|\,\xi_{tT}=x]\right),
\end{align}
so the solution to $\Lambda(t,x)=K$ is
\begin{equation}
    \xi_t^*=\frac{t}{2T}(k_0+k_1)-\frac{T-t}{k_1-k_0} \, \log\left[\frac{p}{1-p}\frac{K-P_{tT}k_0}{P_{tT}k_1-K} \right].
\end{equation}
The price of a call option on $X_{tT}$ is
\begin{equation}
    C_{st}=P_{st}\sum_{i=0}^1 (P_{tT}k-K)\, \Phi\left[\frac{M(k_i)-\xi^*_t}{\sqrt{V}} \right] \Q[X_T=k_i\,|\,\xi_{sT}].
\end{equation}

\section*{Acknowledgements}
The authors are grateful to seminar participants at ETH-Z\"urich, Switzerland, April 2008; at the Bachelier Finance Society Fifth World Congress, London, UK, June 2008; and at the Mathematics in Finance Conference, Kruger National Park, RSA, September 2008, where drafts of this paper were presented, for their comments.

Edward Hoyle acknowledges the support of an EPSRC Doctoral Training Grant and a European Science Foundation research visit grant under the Advanced Mathematical Methods in Finance programme (AMaMeF).

This work was undertaken in part while Edward Hoyle and Lane P.~Hughston were members of the Department of Mathematics, King's College London, and Andrea Macrina was a member of the Department of Mathematics, ETH Z\"urich.
Lane P.~Hughtson acknowledges the support of the Aspen Center for Physics where a portion of this work was carried out.

\nocite{*}
\bibliography{HHM1}

\end{document}